\documentclass{article}
\usepackage{tikz}
\usepackage{blkarray}
\usetikzlibrary{calc}
\usepackage{hyperref}

\input{newsynth.sty}

\title{Efficient Clifford+$T$ approximation of single-qubit operators}

\author{\begin{tabular}{c}
    Peter Selinger \\[.5ex]
    \normalsize Department of Mathematics and Statistics \\
    \normalsize Dalhousie University
  \end{tabular}
}

\date{}

\begin{document}
\maketitle

\begin{abstract}
  We give an efficient randomized algorithm for approximating an
  arbitrary element of $\SU(2)$ by a product of Clifford+$T$
  operators, up to any given error threshold $\epsilon>0$. Under a
  mild hypothesis on the distribution of primes, the algorithm's
  expected runtime is polynomial in $\log(1/\epsilon)$. If the
  operator to be approximated is a $z$-rotation, the resulting gate
  sequence has $T$-count $K+4\log_2(1/\epsilon)$, where $K$ is
  approximately equal to $10$. We also prove a worst-case lower bound
  of $K+4\log_2(1/\epsilon)$, where $K=-9$, so that our algorithm is
  within an additive constant of optimal for certain $z$-rotations.
  For an arbitrary member of $\SU(2)$, we achieve approximations with
  $T$-count $K+12\log_2(1/\epsilon)$.  By contrast, the
  Solovay-Kitaev algorithm achieves $T$-count $O(\log^c(1/\epsilon))$,
  where $c$ is approximately $3.97$.
\end{abstract}

\section{Introduction}

The decomposition of arbitrary unitary operators into gates from some
fixed universal set is a well-known problem in quantum information
theory. If the universal gate set is discrete, the decomposition of a
general operator can only be done approximately, up to a given
accuracy $\epsilon>0$. Here, we focus on the problem of approximating
single-qubit operators using the Clifford+$T$ universal gate set. The
Clifford+$T$ gate set is of particular interest because it is known to
be suitable for fault-tolerant quantum computation {\cite{ZLC2000}}.
Recall that the Clifford group on one qubit is generated by the
Hadamard gate $H$, the phase gate $S$, and the scalar
$\omega=e^{i\pi/4}$. It is well-known that one obtains a universal
gate set by adding the non-Clifford operator $T$.
\begin{equation}
  \omega = e^{i\pi/4}, \quad
  H = \frac{1}{\sqrt2}\zmatrix{cc}{1&1\\1&-1}, \quad
  S = \zmatrix{cc}{1&0\\0&i}, \quad
  T = \zmatrix{cc}{1&0\\0&e^{i\pi/4}}.
\end{equation}

We present an efficient randomized algorithm for approximating an
arbitrary element of $\SU(2)$ by a product of Clifford+$T$ operators,
up to any given error threshold $\epsilon>0$.  Under a mild hypothesis
on the distribution of primes, the algorithm's expected runtime is
polynomial in $\log(1/\epsilon)$. The algorithm approximates any
$z$-rotation with $T$-count{\footnote{We follow {\cite{AMMR12}} in
    using $T$-count, rather than the overall gate count, as a
    convenient measure for the length of a single-qubit Clifford+$T$
    circuit. This is justified, on the one hand, because the
    fault-tolerant implementation of $T$-gates is far more resource
    intensive than that of Clifford gates, and on the other hand,
    because consecutive Clifford gates can always be combined into a
    single Clifford gate, so that the overall gate count is almost
    exactly equal to twice the $T$-count.}}
$K+4\log_2(1/\epsilon)$, where $K$ is
approximately equal to $10$.{\footnote{Throughout the paper, we use
    $K$ to denote an additive constant; difference occurrences of $K$
    can denote different constants.}} We also prove a worst-case lower
bound of
$K+4\log_2(1/\epsilon)$, where $K=-9$, so that our algorithm is
within an additive constant of optimal for $z$-rotations.  For an
arbitrary member of $\SU(2)$, we achieve approximations with $T$-count
$K+12\log_2(1/\epsilon)$.

By contrast, the Solovay-Kitaev algorithm
{\cite{Kitaev97b,KSV2002,Dawson-Nielsen}}, which was until recently
the de facto standard algorithm for this problem, produces circuits of
$T$-count
$O(\log^c(1/\epsilon))$, where $c$ is approximately $3.97$.  Thus, we
have decreased the exponent in the gate complexity from $c=3.97$ to
$c=1$, which is optimal. Moreover, we have decreased the
multiplicative constant to the theoretical optimum, in the case of
worst-case $z$-rotations.

\subsection{Prior work}

Until recently, the state-of-the-art
solution to the approximation problem was the 1995 Solovay-Kitaev algorithm
{\cite{Kitaev97b,KSV2002,Dawson-Nielsen}}. There are various variants
of this algorithm. Perhaps the best-known variant is the one described
in {\cite{Dawson-Nielsen}}, which achieves a gate complexity of
$O(\log^c(1/\epsilon))$, for $c\approx 3.97$. Another well-known
variant is described in Kitaev et al.\ {\cite[Sec.~8.3]{KSV2002}}, and
achieves $c=3+\delta$, where $\delta>0$ is any positive real
number. Kitaev et al.\ also gave another algorithm that uses ancillas
and achieves gate complexity
$O(\log^2(1/\epsilon)\log\log(1/\epsilon))$
{\cite[Sec.~13.7]{KSV2002}}.

At the other end of the spectrum, there is a known
information-theoretic lower bound of $c=1$ for the exponent in the
gate complexity. One can make this lower bound more precise: in fact,
the decomposition of a typical $\SU(2)$ operator into the Clifford+$T$
gate set with accuracy $\epsilon$ requires $T$-count at least
$K+3\log_2(1/\epsilon)$. This follows from a
result of Matsumoto and Amano {\cite{MA08}}. See
Section~\ref{sec-lower-bounds} below for details. Heuristically, it
appears that, for most operators, this lower bound can in fact be
achieved by approximation algorithms based on exhaustive search, such
as Fowler's algorithm {\cite{Fowler04}}. However, the runtime of such
exhaustive search based algorithms is exponential in
$1/\epsilon$.

Recently, Duclos-Cianci and Svore announced an alternative to the
Solovay-Kitaev algorithm that requires ancillas to be prepared in
special resource states, using a state distillation procedure
{\cite{DCS12}}. Using this method, and dependent on the particular
setting, they reduced the gate complexity exponent $c$ to between
$1.12$ and $2.27$. Moreover, the resource states can be prepared
offline, and the expected online gate complexity per single-qubit
operation is constant.

Even more recently, in an important milestone, Kliuchnikov, Maslov,
and Mosca gave an approximation algorithm for single-qubit operators
that has polynomial running time and achieves gate counts of
$O(\log(1/\epsilon))$, thus reducing the gate complexity exponent to
$c=1$ {\cite{KMM-approx}}. The Kliuchnikov-Maslov-Mosca approximation
algorithm is therefore asymptotically optimal. It uses a fixed, small
number of ancilla qubits to approximate a given single-qubit
operator. Unlike approaches based on resource states, the ancillas in
the Kliuchnikov-Maslov-Mosca algorithm are initialized to $\ket{0}$,
and are returned in a state very close (but not exactly equal) to
$\ket{0}$; these ancillas do not require any special preparation
procedure.

For all practical purposes, this use of ancillas in the
Kliuchnikov-Maslov-Mosca algorithm causes no difficulties. However,
the question remained open whether there exists an asymptotically
optimal, efficient single-qubit approximation algorithm that requires
no ancillas, and thus solves exactly the same problem as the
Solovay-Kitaev algorithm, for the Clifford+$T$ gate set. The present
paper achieves such an algorithm.

\subsection{Limitations}

Unlike the Solovay-Kitaev algorithm, which can be applied to any
universal gate set, the algorithm of this paper is specialized to the
Clifford+$T$ gate set. While our number-theoretic approach has already
been generalized to specific other universal gate sets (see, e.g.,
{\cite{KBS13}}), it is unlikely that it would work for {\em arbitrary}
universal gate sets.  While this could be regarded as a limitation of
our method, in practice, the tables will likely be turned: any
universal gate set that does not permit an $O(\log(1/\epsilon))$
efficient synthesis algorithm will probably not be considered
practical in the future.

Technically, the expected polynomial runtime of our algorithm is
contingent on a number-theoretic assumption about the distribution of
primes, as stated below in Hypothesis~\ref{hyp-random}. While this
hypothesis appears to be heuristically true, it has not been proven to
the author's knowledge.

The lower bound of $K+4\log_2(1/\epsilon)$ for the $T$-count of
$\epsilon$-approximations to certain $z$-rotations only applies to the
problem as stated, i.e., for writing operators as a product of
single-qubit Clifford+$T$ operators. By the use of other techniques,
such as ancillas, resource states, or online measurement, even smaller
$T$-counts can be obtained; see, e.g., {\cite{WK13}} for recent
results along these lines.

\section{Overview of the algorithm}\label{sec-summary}

The algorithm can be summarized as follows. Consider the ring
\begin{equation}\label{eqn-domega}
  \D[\omega] = \Z[\frac{1}{\sqrt2},i] =
  \s{\rtt{k}(a\omega^3+b\omega^2+c\omega+d) \mid
    k\in\N; a,b,c,d\in\Z}.
\end{equation}
As shown by Kliuchnikov, Maslov, and Mosca {\cite{Kliuchnikov-etal}},
a unitary operator $U\in\U(2)$ can be exactly synthesized over the
Clifford+$T$ gate set (with no ancillas) if and only if all the matrix
entries belong to the ring $\D[\omega]$. Moreover, the required number
of $T$-gates is at most $2k$, where $k>0$ is the minimal exponent
required to write all entries of $U$ in the form mentioned in
(\ref{eqn-domega}).

Suppose we wish to approximate a given $z$-rotation
\begin{equation}
  R_z(\theta) = e^{-i\theta Z/2} = \zmatrix{cc}{e^{-i\theta/2} & 0 \\ 0 & e^{i\theta/2}}
\end{equation}
up to a given $\epsilon>0$, using Clifford+$T$ gates. We will choose
an integer $k$ and a randomized sequence of suitable elements
$u\in\Z[\omega]$ such that $\frac{u}{\rt{k}}\approx
e^{-i\theta/2}$. For each $u$, we attempt to solve the Diophantine
equation
\begin{equation}\label{eqn-dioph1}
  t\da t + u\da u = 2^k,
\end{equation}
with $t\in\Z[\omega]$. The parameters are chosen in such a way that
this succeeds for a relatively large proportion of the available
$u$. This yields a unitary matrix
\begin{equation}
  U = \rtt{k}\zmatrix{cc}{u & -t\da \\ t & u\da}
\end{equation}
with $\norm{U-R_z(\theta)}\leq \epsilon$, and such that the
coefficients of $U$ are in the ring $\D[\omega]$. By the
Kliuchnikov-Maslov-Mosca exact synthesis algorithm, $U$ can be exactly
decomposed into Clifford+$T$ gates with $T$-count at most $2k$. The
remainder of this paper fills in the details of these ideas: in
particular, how to choose $k$, how to find ``suitable'' $u$, and how
to solve the Diophantine equation (\ref{eqn-dioph1}) with high
probability.

\section{Some number theory}

In this section and the next one, we summarize some well-known facts
from algebraic number theory. The following exposition requires almost
no prerequisites, and will hopefully be of use to readers who are not
experts in number theory, or to those who wish to implement the
algorithm of this paper.

\subsection{Some rings of algebraic integers}

Recall that $\N$ is the set of natural numbers including 0. Let
$\omega=e^{i\pi/4} = (1+i)/\sqrt2$. Note that $\omega$ is an 8th
root of unity satisfying $\omega^2=i$ and $\omega^4=-1$. We will
consider the following rings:
\begin{itemize}
\item $\Z$, the ring of integers;
\item $\Z[\sqrt2] = \s{a+b\sqrt2 \mid a,b\in\Z}$, the ring of {\em
    quadratic integers with radicand 2};
\item $\Z[i] = \s{a+bi \mid a,b\in\Z}$, the ring of {\em Gaussian
    integers};
\item $\Z[\omega] = \s{a\omega^3+b\omega^2+c\omega+d \mid
    a,b,c,d\in\Z}$, the ring of {\em cyclotomic integers of degree 8}.
\end{itemize}

\begin{remark}
  We have the inclusions $\Z \seq \Z[\sqrt2] \seq \Z[\omega]$ and
  $\Z \seq \Z[i] \seq \Z[\omega]$. Of course, all four rings are
  subrings of the complex numbers.
\end{remark}

\subsection{Conjugate and norm}

\begin{definition}[Conjugation]
  Since $\omega$ is a root of the irreducible polynomial $x^4+1$, the
  ring $\Z[\omega]$ has four automorphisms. One of these is the usual
  {\em complex conjugation}, which we denote $(-)\da$. It maps $i$ to
  $-i$, and $\sqrt2$ to itself. Equivalently, it is given by
  $\omega\da=-\omega^3$. It acts trivially on $\Z$ and $\Z[\sqrt2]$,
  and non-trivially on $\Z[\omega]$ and $\Z[i]$, with the following
  explicit formulas, for real $a,b,c,d$:
  \begin{align}
    (a\omega^3+b\omega^2+c\omega+d)\da &= -c\omega^3-b\omega^2-a\omega+d,\\
    (a + bi)\da &= a - bi.
  \end{align}
  Another automorphism is {\em $\sqrt2$-conjugation}, which we denote
  $(-)\bul$. It maps $\sqrt2$ to $-\sqrt2$, and $i$ to
  itself. Equivalently, $\omega\bul=-\omega$. It acts trivially on $\Z$
  and $\Z[i]$, and non-trivially on $\Z[\omega]$ and $\Z[\sqrt2]$,
  with the following explicit formulas, for rational $a,b,c,d$:
  \begin{align}
    (a\omega^3+b\omega^2+c\omega+d)\bul &= -a\omega^3+b\omega^2-c\omega+d,\\
    (a + b\sqrt2)\bul &= a - b\sqrt2.
  \end{align}
  The remaining two automorphisms are, of course, the identity
  function and $(-)\dabul=(-)\bulda$.
\end{definition}

\begin{remark}\label{rem-z-omega-subrings}
  For $t\in\Z[\omega]$, we have $t\in\Z[\sqrt2]$ iff $t=t\da$,
  $t\in\Z[i]$ iff $t=t\bul$, and $t\in\Z$ iff $t=t\da$ and $t=t\bul$.
\end{remark}

\begin{definition}[Norms]
  We define an integer-valued (number-theoretic) {\em norm} on each ring:
  \begin{itemize}
  \item For $t=a+bi\in\Z[i]$, let
    \begin{equation}\label{eqn-iN}
      \iN(t) = t\da t = a^2 + b^2.
    \end{equation}
  \item For $t=a+b\sqrt2\in\Z[\sqrt2]$, let
    \begin{equation}\label{eqn-rN}
      \rN(t) = t\bul t = a^2 - 2b^2.
    \end{equation}
  \item For $t=a\omega^3+b\omega^2+c\omega+d\in\Z[\omega]$, let
    \begin{equation}\label{eqn-yN}
      \yN(t) = (t\da t)\bul (t\da t) = (a^2+b^2+c^2+d^2)^2-2(ab+bc+cd-da)^2.
    \end{equation}
  \end{itemize}
  For consistency, we also define $\zN(t) = t$ for $t\in\Z$.
\end{definition}

\begin{remark}\label{rem-norm}
  $\iN$ and $\yN$ are valued in the non-negative integers, but $\zN$
  and $\rN$ may take negative values. Each norm is multiplicative, in
  the sense that $\nN(ts)=\nN(t)\nN(s)$ for all $t,s$. Moreover,
  $\nN(t)=0$ if and only if $t=0$, and $\nN(t)=\pm 1$ if and only if
  $t$ is a unit (i.e., an invertible element) in the ring. The latter
  property follows easily from multiplicativity and the fact that
  equations (\ref{eqn-iN}--\ref{eqn-yN}) define an inverse for $t$
  when $\nN(t)=\pm 1$.
\end{remark}

\subsection{Euclidean domains}

\begin{remark}\label{rem-euclid}
  $\Z$, $\Z[\sqrt2]$, $\Z[i]$, and $\Z[\omega]$ are Euclidean
  domains. Explicitly, for each of these rings, a Euclidean function
  is given by $|\nN(t)|$. For given elements $s$ and $t\neq 0$, the
  division of $s$ by $t$ with quotient $q$ and remainder $r$ can be
  defined by first calculating $q'=s/t$ in the corresponding field of
  fractions (respectively $\Q$, $\Q[\sqrt2]$, $\Q[i]$, and
  $\Q[\omega]$). $q'$ can be expressed with rational coefficients as
  $a$, $a+b\sqrt2$, $a+bi$, or $\omega^3a+\omega^2b+\omega c+d$,
  respectively. Then $q$ is obtained from $q'$ by rounding each
  coefficient $a,b,c,d$ to the closest integer, and $r$ is defined to
  be $qt-s$. One may verify that in each case, $|\nN(r)|\leq
  \frac{9}{16}|\nN(t)|$.
\end{remark}

As usual, we write $t\divides s$ to mean that $t$ is a divisor of $s$,
i.e., that there exists some $r$ such that $rt=s$.  We also write
$t\sim s$ to indicate that $t\divides s$ and $s\divides t$; in this
case, $t$ and $s$ differ only by a unit of the ring, and we say $t$
and $s$ are {\em associates}.

\begin{remark}\label{rem-divisibility} 
  We note that if $R$ is one of the four rings $\Z$, $\Z[\sqrt2]$,
  $\Z[i]$, and $\Z[\omega]$, and $rt=s$ for some $s, t\in R$ and
  $r\in\Z[\omega]$ with $t\neq 0$, then $r\in R$. This is easily
  proved by letting $\overline R$ be the corresponding field of
  fractions (respectively $\Q$, $\Q[\sqrt2]$, $\Q[i]$, and
  $\Q[\omega]$), and noting on the one hand that
  $r=\frac{s}{t}\in\overline R$, and on the other hand that
  $R=\Z[\omega]\cap \overline R$. The latter property follows from
  Remark~\ref{rem-z-omega-subrings}.

  In particular, it follows that $t\divides s$ holds in $R$ if and only if
  $t\divides s$ holds in $\Z[\omega]$, so that the notion of divisibility is
  independent of $R$.
\end{remark}

\begin{remark}\label{rem-efficient-gcd}
  Every Euclidean domain admits greatest common divisors, which can be
  calculated by repeated divisions with remainder via Euclid's
  algorithm.  Note that greatest common divisors are only unique up to
  a unit of the ring. Also note that, since each division by $t$ with
  remainder $r$ satisfies $|\nN(r)|\leq \frac{9}{16}|\nN(t)|$, the
  computation of the greatest common divisor of two elements of any of
  the above rings requires at most $O(\log |\nN(t)|)$ divisions with
  remainder.
\end{remark}

An element $t$ of a Euclidean domain is called {\em prime} (or {\em
  irreducible}) if $t$ is not a unit, and for all $r,s$ with $rs=t$,
either $r$ or $s$ is a unit. We note that the notion of primality is
not independent of the ring. For example, $7$ is prime in $\Z$, but
not in $\Z[\sqrt2]$, as $7=(3+\sqrt2)(3-\sqrt2)$.

\begin{remark}\label{rem-prime}
  In each of the above rings, if $\nN(t)$ is prime in $\Z$, then $t$
  is prime in the ring. Indeed, suppose $t=rs$. Since
  $\nN(t)=\nN(r)\nN(s)$, either $\nN(r)$ or $\nN(s)$ is $\pm1$, hence
  $r$ or $s$ is a unit by Remark~\ref{rem-norm}.
\end{remark}

\subsection[{Units in Z[root 2]}]{Units in $\Z[\sqrt2]$}

\begin{lemma}\label{lem-units}
  The units of $\Z[\sqrt2]$ are exactly the elements of the form
  $(-1)^n(\sqrt2-1)^k$, where $n,k\in\Z$. A unit $u$ is a square if
  and only if $u\geq 0$ and $u\bul\geq 0$.
\end{lemma}

\begin{proof}
  We first note that $(\sqrt2-1)(\sqrt2+1)=1$, so $\sqrt2-1$ and
  $\sqrt2+1$ are units. Now consider any unit
  $u=a+b\sqrt2\in\Z[\sqrt2]$. We prove the first claim by induction on
  $|b|$. The base case is $b=0$; in this case, $\rN(u)=a^2=\pm1$
  implies $u=a=\pm1$. For the induction step, note that
  $\rN(u)=a^2-2b^2=\pm1$ implies $a\neq 0$ and $a^2<2b^2+2\leq 4b^2$,
  hence $0 < |a| < |2b|$.  First consider the case where $a,b$ have
  the same sign. Then $|a-b| < |b|$, and the claim is proved by
  applying the induction hypothesis to $\uhat=u(\sqrt2-1) = 2b-a +
  (a-b)\sqrt2$. The case where $a,b$ have opposite signs is similar,
  except we use $|a+b|<|b|$ to apply the induction hypothesis to
  $\uhat=u/(\sqrt2-1) = a+2b + (a+b)\sqrt2$.

  For the second claim, note that $u=(-1)^n(\sqrt2-1)^k$ satisfies
  $u\geq 0$ iff $n$ is even, and $u\bul\geq 0$ iff $n+k$ is
  even. Therefore $u$ is a square iff $n$ and $k$ are both even, iff
  $u\geq 0$ and $u\bul\geq 0$.
\end{proof}

\subsection{Roots of $-1$ in $\Z/(p)$}

\begin{remark}\label{rem-h}
  Let $p\in\Z$ be a positive prime satisfying $p\equiv 1\,(\mymod 4)$.
  It is well-known that there exists $h\in\Z$ such that
  $h^2+1\equiv0\mmod{p}$. We recall that there is an efficient randomized
  algorithm for computing $h$. Consider the field $\Z/(p)$ of integers
  modulo $p$. By Fermat's Little Theorem, all non-zero $b\in\Z/(p)$
  satisfy $b^{p-1}=1$, hence $b^{(p-1)/2}=\pm1$. Because each of the
  polynomial equations $b^{(p-1)/2}=1$ and $b^{(p-1)/2}=-1$ has at
  most $(p-1)/2$ solutions, $b^{(p-1)/2}=-1$ holds for exactly half of
  the non-zero $b\in\Z/(p)$. Therefore, one can efficiently solve
  $b^{(p-1)/2}=-1$ by picking $b$ at random until a solution is found;
  on average, this will require two attempts.  Note that, by the
  method of repeated squaring, the computation of $b^{(p-1)/2}$ only
  requires $O(\log p)$ multiplications. Finally, we can set
  $h=b^{(p-1)/4}$.
\end{remark}

\section{A Diophantine equation}\label{sec-dioph}

We will be interested in solving equations of the form
\begin{equation}\label{eqn-xi}
  t\da t = \xi,
\end{equation}
where $\xi\in\Z[\sqrt2]$ is given and $t\in\Z[\omega]$ is
unknown. Number theorists will recognize this as a relative norm
equation, which can be solved by splitting fully split primes in
$\Z[\omega]$ {\cite{Cohen}}. In the interest of self-containedness,
and to aid in the complexity analysis of Section~\ref{sec-complexity},
we describe a method for solving equation (\ref{eqn-xi}) in detail.

Writing $\xi=x+y\sqrt2$ and $t=a\omega^3+b\omega^2+c\omega+d$, we can
equivalently express (\ref{eqn-xi}) as a pair of integer equations:
\begin{align}
  a^2+b^2+c^2+d^2  &~=~ x\\
  ab+bc+cd-da &~=~ y.\label{eqn-abcdy}
\end{align}
Of course, not every $\xi\in\Z[\sqrt2]$ can be expressed in the form
(\ref{eqn-xi}). However, we have the following:

\begin{theorem}\label{thm-dioph}
  Let $\xi=x+y\sqrt2\in\Z[\sqrt2]$, where $x$ is odd, $y$ is even,
  $\xi\geq 0$, $\xi\bul\geq 0$, and $p=\xi\bul\xi=x^2-2y^2$ is prime in
  $\Z$. Then there exists $t\in\Z[\omega]$ satisfying
  (\ref{eqn-xi}). Moreover, there is an efficient randomized algorithm
  for computing $t$.
\end{theorem}

\begin{remark}\label{rem-positive-new}
  Since $t\da t=\xi$ implies $t\bulda t\bul = \xi\bul$, the conditions
  $\xi\geq 0$ and  $\xi\bul\geq 0$ are both obviously necessary for
  (\ref{eqn-xi}) to have a solution.
\end{remark}

\begin{proof}[Proof of Theorem~\ref{thm-dioph}]
  Note that $p$ is prime by assumption. Since $\xi\geq 0$ and
  $\xi\bul\geq 0$, we know that $p\geq 0$. Since $x$ is odd and $y$ is
  even, we have $p\equiv 1\,(\mymod 4)$. Therefore, by
  Remark~{\ref{rem-h}}, we can find $h\in\Z$ with $p\divides
  h^2+1$. Moreover, in $\Z[\sqrt2]$, we have $\xi\divides p$ and
  therefore $\xi\divides h^2+1$.  Let $s=\gcd(h+i,\xi)$ in the ring
  $\Z[\omega]$. We claim that $s\da s\sim\xi$.

  First, note that $\xi$ is prime in the ring $\Z[\sqrt2]$ by
  Remark~\ref{rem-prime}. By definition of $s$, we know that
  $s\divides \xi$. But $\xi$ is real, and therefore also $s\da\divides
  \xi$. It follows that $s\da s\divides \xi^2$ in $\Z[\omega]$. By
  Remark~\ref{rem-divisibility}, $s\da s\divides \xi^2$ in
  $\Z[\sqrt2]$. Since $\xi$ is prime in $\Z[\sqrt2]$, it follows that
  one of three cases holds:
  \begin{enumerate}\alphalabels
  \item $s\da s\sim 1$. But this is impossible. Indeed, in
    this case, $s$ is a unit, so $h+i$ and $\xi$ are relatively
    prime. As $\xi$ is real, it follows that $h-i$ and $\xi$ are also
    relatively prime, hence $(h+i)(h-i)$ is relatively prime to $\xi$,
    contradicting $\xi\divides h^2+1$.
  \item $s\da s\sim\xi$. This is what was claimed.
  \item $s\da s\sim\xi^2$. This is also impossible. Indeed, in this
    case, $\yN(s)=\rN(s\da s)=\rN(\xi^2)=\yN(\xi)$. But we also have
    $s\divides \xi$, so $\xi=us$ for some $u\in\Z[\omega]$. Therefore
    $\yN(u)=\yN(\xi)/\yN(s)=1$, hence $u$ is a unit in $\Z[\omega]$,
    thus $s\sim\xi$. Since, by definition, $s\divides h+i$, we have
    $\xi\divides h+i$. Since $h$ is an integer, we have $h=h\bul$, and
    thus $\xi\bul\divides h+i$. We note that $\xi\not\sim\xi\bul$, for
    otherwise, we would have $p=\xi\bul\xi\divides \xi^2+{\xi\bul}^2 =
    2x^2+4y^2=4x^2-2p$; since $p$ is an odd prime, this implies
    $p\divides x$, and from $p=x^2-2y^2$, we get $p\divides y$, hence
    $p^2\divides x^2-2y^2= p$, a contradiction. Therefore,
    $\xi$ and $\xi\bul$ are non-associate primes in $\Z[\sqrt2]$, both
    dividing $h+i$, which implies $p=\xi\bul\xi\divides h+i$, an
    absurdity since $\frac{1}{p}(h+i)\not\in\Z[\omega]$.
  \end{enumerate}
  We have shown the existence of $s\in\Z[\omega]$ such that
  $s\da s\sim\xi$ in $\Z[\sqrt2]$. Let $u$ be a unit of $\Z[\sqrt2]$
  such that $us\da s=\xi$. Our next claim is that $u$ is a square in
  $\Z[\sqrt2]$. First note that, by the usual properties of complex
  numbers, $s\da s\geq 0$ and $(s\da s)\bul = (s\bul)\da(s\bul)\geq
  0$. Also, $\xi\geq 0$ and $\xi\bul\geq 0$ by assumption.  It follows
  that $u\geq 0$ and $u\bul\geq 0$. But then $u$ is a square by
  Lemma~\ref{lem-units}. Let $v\in\Z[\sqrt2]$ with $v^2=u$, and let
  $t=vs$. Noting that $v=v\da$, we have $t\da t=v^2 s\da s = \xi$, as
  desired.
\end{proof}

\begin{remark}
  We note that the proof of Theorem~\ref{thm-dioph} immediately yields
  an algorithm for computing $t$. As a matter of fact, the only
  randomized aspect is the computation of $h$; the remainder consists
  of arithmetic calculations in the various rings, and can be done
  deterministically. The computation of $s$ requires taking a greatest
  common divisor in $\Z[\omega]$, which can be done efficiently by
  Remark~\ref{rem-efficient-gcd}. The computation of $u$ requires a
  simple division in $\Z[\sqrt2]$, and the computation of $v$ requires
  taking a square root in $\Z[\sqrt2]$, which easily reduces to
  solving a quadratic equation in the integers. Finally, the
  computation of $t$ requires a multiplication in $\Z[\omega]$.
\end{remark}

\begin{remark}
  Theorem~\ref{thm-dioph} is analogous to a well-known theorem about
  the integers, stating that every positive prime satisfying
  $p\equiv 1\,(\mymod 4)$ can be written as a sum of two squares
  $p=a^2+b^2$, or equivalently, that the equation $z\da z = p$ has a solution
  $z=a+bi\in\Z[i]$. A randomized algorithm for computing $z$ was
  described, for example, by Rabin and Shallit {\cite{Rabin-Shallit}}.
  Our proof and algorithm follows the same general idea, applied to a
  different pair of Euclidean rings.
\end{remark}

\section[{Approximations in Z[root 2]}]{Approximations in $\Z[\sqrt2]$}

It is well-known that the set $\Z[\sqrt2]$ of integers of the form
$\alpha=a+b\sqrt2$ is a dense subset of the real numbers. Here, density is of
course understood with respect to the Euclidean distance
$|\alpha-\beta|$. We note that the Euclidean distance is not at all
preserved by $\sqrt2$-conjugation; in fact, as we will see in
Lemma~\ref{lem-interval-lower} below, unless $\alpha=\beta$, it is
impossible for $|\alpha-\beta|$ and $|\alpha\bul-\beta\bul|$ to be small
at the same time.

The purpose of this section is to find solutions in $\Z[\sqrt2]$ to
constraints involving $\alpha$ and $\alpha\bul$
simultaneously. Specifically, we will be interested in solving
problems of the form
\begin{equation}\label{eqn-approx-interval}
  a+b\sqrt2 \in [x_0,x_1]\quad\mbox{and}\quad a-b\sqrt2 \in [y_0,y_1],
\end{equation}
where $x_0<x_1$ and $y_0<y_1$ are given real numbers, and $a,b$ are
unknown integers. We start with a result limiting the number of
solutions.

\begin{lemma}\label{lem-interval-lower}
  Let $[x_0,x_1]$ and $[y_0,y_1]$ be closed intervals of real
  numbers. Let $\delta = x_1-x_0$ and $\Delta=y_1-y_0$, and assume
  $\delta\Delta < 1$. Then there exists at most one
  $\alpha=a+b\sqrt2\in\Z[\sqrt2]$ satisfying
  (\ref{eqn-approx-interval}).
\end{lemma}

\begin{proof}
  Suppose $\alpha$ and $\beta$ are two solutions. Then
  \begin{equation}
    |\rN(\alpha-\beta)| =
    |\alpha-\beta|\cdot|\alpha\bul-\beta\bul| \leq \delta\Delta < 1.
  \end{equation}
  Since $\rN(\alpha-\beta)$ is an integer, it follows that
  $\rN(\alpha-\beta)=0$, and therefore $\alpha=\beta$.
\end{proof}

The next result establishes the existence of solutions.

\begin{lemma}\label{lem-interval}
  Let $[x_0,x_1]$ and $[y_0,y_1]$ be closed intervals of real
  numbers. Let $\delta = x_1-x_0$ and $\Delta=y_1-y_0$, and assume
  $\delta\Delta \geq (1+\sqrt2)^2$. Then there exists at least one
  $\alpha=a+b\sqrt2\in\Z[\sqrt2]$ satisfying
  (\ref{eqn-approx-interval}). Moreover, there is an efficient
  algorithm for computing such $a$ and $b$.
\end{lemma}

\begin{proof}
  Let us say that a pair of positive real numbers $(\delta,\Delta)$
  has the {\em coverage property} if for all $x,y\in\R$, there exists
  $\alpha\in\Z[\sqrt2]$ with
  \begin{equation}\label{eqn-grideq}
    (\alpha,\alpha\bul)\in[x,x+\delta]\times[y,y+\Delta].
  \end{equation}
  The goal, then, is to show that $\delta\Delta \geq (1+\sqrt2)^2$
  implies the coverage property.

  Before we continue, it is perhaps helpful to consider the following
  illustration. Here, $a$ is shown on the horizontal axis, and
  $b\sqrt2$ is shown on the vertical axis. The region defined by
  $(a+b\sqrt2,a-b\sqrt2)\in[x,x+\delta]\times[y,y+\Delta]$ is a
  rectangle oriented at 45 degrees. We are only interested in solutions
  where $a,b$ are integers; these solutions therefore lie on the grid
  $\Z\times\sqrt2\,\Z$. Note that the horizontal and vertical spacing
  are not the same.
  
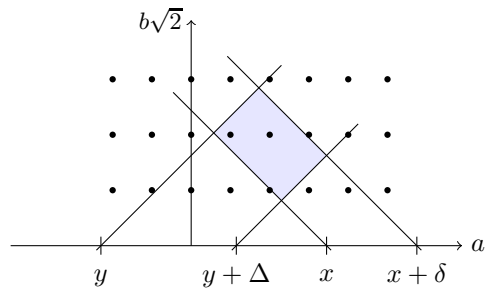
\[ \begin{tikzpicture}[scale=0.6]
    \fill[fill=blue!10] (2,1) -- (3,2) -- (1.5,3.5) -- (0.5,2.5) -- cycle;
    \foreach\x in {-2,...,5} {
      \foreach\y in {1,...,3} {
        \fill (.87*\x, 1.23*\y) circle (.07);
      }
    }
    \draw[->] (-4,0) -- (6,0) node[right] {$a$};
    \draw[->] (0,0) -- (0,5) node[left] {\small$b\sqrt2$};
    \draw (-2.1,-0.1) -- (2,4);
    \draw (0.9,-0.1) -- (3.7,2.7);
    \draw (3.1,-0.1) -- (-.4,3.4);
    \draw (5.1,-0.1) -- (0.8,4.2);
    \draw (-2,0.2) -- (-2,-0.2) node[below] {\rule{0mm}{1.5ex}$y$};
    \draw (1,0.2) -- (1,-0.2) node[below] {\rule{0mm}{1.5ex}$y+\Delta$};
    \draw (3,0.2) -- (3,-0.2) node[below] {\rule{0mm}{1.5ex}$x$};
    \draw (5,0.2) -- (5,-0.2) node[below] {\rule{0mm}{1.5ex}$x+\delta$};
  \end{tikzpicture}
  \]
  We prove a sequence of claims leading up to the lemma.
  \begin{enumerate}\alphalabels
  \item[(a)] If $(\delta,\Delta)$ has the coverage property and
    $\delta'\geq\delta$, $\Delta'\geq\Delta$, then $(\delta',\Delta')$
    has the coverage property. This is a trivial consequence of the
    definitions.
  \item[(b)] For all $\delta,\Delta>0$, the pair $(\delta,\Delta)$ has
    the coverage property if and only if $(\Delta,\delta)$ has the
    coverage property. This is a trivial consequence of the definitions.
  \item[(c)] Let $\lambda=1+\sqrt2$. Then for all $\delta,\Delta>0$,
    the pair $(\delta,\Delta)$ has the coverage property if and only
    if $(\lambda\delta,\lambda\inv\Delta)$ has the coverage property.
    Indeed, recall that $\lambda\bul=-\lambda\inv=1-\sqrt2$. Therefore
    both $\lambda$ and $\lambda\inv$ are elements of $\Z[\sqrt2]$.
    Suppose $(\delta,\Delta)$ has the coverage property, and
    $x,y\in\R$ are given. Then there exists $\alpha\in\Z[\sqrt2]$ with
    $(\alpha,\alpha\bul)\in[\lambda\inv x,\lambda\inv
    x+\delta]\times[-\lambda y-\Delta,-\lambda y]$. Let
    $\beta=\lambda\alpha$. It follows that
    \begin{equation}
      (\beta,\beta\bul) = 
      (\lambda\alpha,\lambda\bul\alpha\bul) =
      (\lambda\alpha,-\lambda\inv\alpha\bul)
      \in [x,x+\lambda\delta]\times[y,y+\lambda\inv\Delta],
    \end{equation}
    so $(\lambda\delta,\lambda\inv\Delta)$ has the coverage property.
    The converse if proved symmetrically.
  \item[(d)] $(\delta,\Delta) = (1+\sqrt2,\sqrt2)$ has the coverage
    property. Geometrically, this is equivalent to the statement that
    $\R^2$ is covered by all translations along the grid
    $\Z\times\sqrt2\,\Z$ of the rectangle $R$ defined by
    $(a+b\sqrt2,a-b\sqrt2)\in[0,1+\sqrt2]\times[0,\sqrt2]$.
    \[ \begin{tikzpicture}[scale=0.6]
      \foreach\x in {3,2,...,0} {
        \foreach\y in {2,1,...,0} {
          \filldraw[fill=blue!10,draw=black,shift={(\x,1.414*\y)}] (0,0) -- (.707,-.707) --
          (1.914,0.5) -- (1.207,1.207) -- cycle;
        }
      }
      \filldraw[fill=red!10,draw=black] (0,0) -- (.707,-.707) --
      (1.914,0.5) -- (1.207,1.207) -- cycle;
      \foreach\x in {-1,...,5} {
        \foreach\y in {0,...,2} {
          \fill (\x, 1.414*\y) circle (.07);
        }
      }
      \draw[->] (-2.3,0) -- (6.3,0) node[right] {$a$};
      \draw[->] (0,0) -- (0,4.2) node[left] {\small$b\sqrt2$};
    \end{tikzpicture}
    \]
    In order to have an explicit algorithm for finding $\alpha$, we give
    an algebraic proof. Consider arbitrary $x,y\in\R$. Let $a,b$ be
    integers such that 
    \begin{align}
      a-1 &~\leq~ \frac{x+y+\Delta}{2} ~<~ a,\label{eqn-i1} \\
      (b-1)\sqrt2 & ~\leq~ \frac{x-y-\Delta}{2} ~<~ b\sqrt2.\label{eqn-i2}
    \end{align}
    Let $\alpha=a+b\sqrt2$, $\alpha'=a+(b+1)\sqrt2$, and
    $\alpha''=(a-1)+b\sqrt2$. We claim that either $\alpha$,
    $\alpha'$, or $\alpha''$ is a solution to
    (\ref{eqn-grideq}).
    \begin{itemize}
    \item Case 1: $a-b\sqrt2\leq y+\Delta$. In this case, we have:
      \begin{align*}
        \alpha &= a+b\sqrt2 > \frac{x+y+\Delta}{2} +
        \frac{x-y-\Delta}{2} = x
        &&\mbox{by (\ref{eqn-i1}) and (\ref{eqn-i2})},\\
        \alpha &= a+b\sqrt2 \leq \frac{x+y+\Delta}{2} + 1 +
        \frac{x-y-\Delta}{2} + \sqrt2 = x+\delta
        &&\mbox{by (\ref{eqn-i1}) and (\ref{eqn-i2})},\\
        \alpha\bul &= a-b\sqrt2 > \frac{x+y+\Delta}{2} -
        \frac{x-y-\Delta}{2}-\sqrt2 = y
        &&\mbox{by (\ref{eqn-i1}) and (\ref{eqn-i2})},\\
        \alpha\bul &= a-b\sqrt2 \leq y+\Delta
        &&\mbox{by assumption}.
      \end{align*}
      Therefore, $\alpha$ is a solution to
      (\ref{eqn-grideq}).
    \item Case 2: $a-b\sqrt2 > y+\Delta$ and $a+b\sqrt2 \leq x+1$. In
      this case, we have:
      \begin{align*}
        \alpha' &= a+(b+1)\sqrt2 > \frac{x+y+\Delta}{2} +
        \frac{x-y-\Delta}{2} + \sqrt2 = x + \sqrt2 > x
        &&\mbox{by (\ref{eqn-i1}) and (\ref{eqn-i2})},\\
        \alpha' &= a+(b+1)\sqrt2 \leq x+1+\sqrt2 = x+\delta
        &&\mbox{by assumption},\\
        {\alpha'}\bul &= a-(b+1)\sqrt2 > y+\Delta-\sqrt2 = y
        &&\mbox{by assumption},\\
        {\alpha'}\bul &= a-(b+1)\sqrt2 < \frac{x+y+\Delta}{2} + 1 -
        \frac{x-y-\Delta}{2} - \sqrt2 = y+\Delta+1-\sqrt2 < y+\Delta
        &&\mbox{by (\ref{eqn-i1}) and (\ref{eqn-i2})}.
      \end{align*}
      Therefore, $\alpha'$ is a solution to
      (\ref{eqn-grideq}).
    \item Case 3: $a-b\sqrt2 > y+\Delta$ and $a+b\sqrt2 > x+1$. In
      this case, we have:
      \begin{align*}
        \alpha'' &= (a-1)+b\sqrt2 > x
        &&\mbox{by assumption},\\
        \alpha'' &= (a-1)+b\sqrt2 \leq \frac{x+y+\Delta}{2} +
        \frac{x-y-\Delta}{2} + \sqrt2 = x+\sqrt2 < x+\delta
        &&\mbox{by (\ref{eqn-i1}) and (\ref{eqn-i2})},\\
        {\alpha''}\bul &= (a-1)-b\sqrt2 > y+\Delta-1 > y
        &&\mbox{by assumption},\\
        {\alpha''}\bul &= (a-1)-b\sqrt2 \leq
        \frac{x+y+\Delta}{2} - \frac{x-y-\Delta}{2} = y+\Delta
        &&\mbox{by (\ref{eqn-i1}) and (\ref{eqn-i2})}.
      \end{align*}
      Therefore, $\alpha''$ is a solution to
      (\ref{eqn-grideq}).
    \end{itemize}
  \item[(e)] $(\delta,\Delta) = (2+\sqrt2, 1)$ has the coverage
    property.  This follows from (b)--(d), by noting that
    $2+\sqrt2=\lambda \sqrt2$ and $1=\lambda\inv (1+\sqrt2)$. 
  \item[(f)] Suppose $\delta\Delta\geq (1+\sqrt2)^2$ and $1<\delta
    \leq 1+\sqrt2$. Then $(\delta,\Delta)$ has the coverage property.
    We consider two cases. Case 1: $\delta > \sqrt2$. Since
    $\delta\leq 1+\sqrt2$, we have $\Delta\geq 1+\sqrt2$. But
    $(\sqrt2,1+\sqrt2)$ has the coverage property by (d), so the claim
    follows from (a). Case 2: $\delta\leq\sqrt2$. Then
    \begin{equation}
      \Delta \geq \frac{(1+\sqrt2)^2}{\sqrt2} > 2+\sqrt2.
    \end{equation}
    Also $\delta>1$. But $(1,2+\sqrt2)$ has the coverage property by
    (e), so the claim follows from (a).
  \item[(g)] Suppose $\delta,\Delta>0$ such that $\delta\Delta\geq
    (1+\sqrt2)^2$. Then $(\delta,\Delta)$ has the coverage property.
    Indeed, note that there exists some $n\in\Z$ such that
    $1<\lambda^{n}\delta\leq \lambda = 1+\sqrt2$. Then
    $(\lambda^{n}\delta,\lambda^{-n}\Delta)$ has the coverage property
    by (f), and $(\delta,\Delta)$ has the coverage property by (c).
  \end{enumerate}
  This finishes the proof of the lemma. Note that in each step of the
  proof, we have shown explicitly how to compute a solution
  $\alpha$. Since there are no iterative constructions (apart from the
  computation of $\lambda^n$ in (g), which can be done in
  $O(|\log(\delta)|)$ steps), the computational effort is essentially
  trivial, and certainly not greater than say $O(\log^2(\delta))$.
  \qedhere
\end{proof}

\begin{remark}
  We note that the bounds on $\delta\Delta$ in
  Lemmas~\ref{lem-interval-lower} and {\ref{lem-interval}} are
  sharp. For the sharpness of Lemma~\ref{lem-interval-lower}, consider
  $(\alpha,\alpha\bul)\in[0,1]\times[0,1]$, which has two solutions
  $\alpha=0$ and $\alpha=1$. For the sharpness of Lemma~\ref{lem-interval},
  consider $(\alpha,\alpha\bul)\in[0,1+\sqrt2]\times[-\sqrt2,1]$, which
  has exactly four solutions
  $\alpha=0$, $\alpha=1$, $\alpha=\sqrt2$, and $\alpha=1+\sqrt2$. Therefore, for all $\epsilon>0$,
  $(\alpha,\alpha\bul)\in[\epsilon,1+\sqrt2-\epsilon]\times[-\sqrt2+\epsilon,1-\epsilon]$
  has no solutions at all.
\end{remark}

We also give a version of Lemma~\ref{lem-interval} where $a$ is
restricted to be either even or odd:

\begin{corollary}\label{cor-interval-even-odd}
  Let $[x_0,x_1]$ and $[y_0,y_1]$ be closed intervals of real
  numbers. Let $\delta = x_1-x_0$ and $\Delta=y_1-y_0$, and assume
  $\delta\Delta \geq 2(1+\sqrt2)^2$. Then there exist $a',b'\in\Z$
  such that $a'+b'\sqrt2\in[x_0,x_1]$ and $a'-b'\sqrt2 \in [y_0,y_1]$,
  and $a'$ is even. There also exist $a'',b''\in\Z$ such that
  $a''+b''\sqrt2\in[x_0,x_1]$ and $a''-b''\sqrt2 \in [y_0,y_1]$, and
  $a''$ is odd. Moreover, there is an efficient algorithm for
  computing such $a'$, $b'$, $a''$, and $b''$.
\end{corollary}

\begin{proof}
  This is proved by rescaling. To prove the first claim, use
  Lemma~\ref{lem-interval} to find $a,b\in\Z$ with
  $a+b\sqrt2\in[x_0/\sqrt2,x_1/\sqrt2]$ and $a-b\sqrt2 \in
  [-y_1/\sqrt2,-y_0/\sqrt2]$. Let $a'=2b$ and $b'=a$. Then we have
  $a'+b'\sqrt2 = \sqrt2\,(a+b\sqrt2)\in [x_0,x_1]$ and $a'-b'\sqrt2 =
  -\sqrt2\,(a-b\sqrt2) \in [y_0,y_1]$, as desired.  To prove the
  second claim, use the first claim to find $a'+b'\sqrt2 \in
  [x_0-1,x_1-1]$ and $a'-b'\sqrt2\in [y_0-1,y_1-1]$, with $a'$ even;
  then let $a''=a'+1$ and $b''=b'$.
\end{proof}

\section[Approximation up to epsilon]{Approximation up to $\epsilon$}

As mentioned in Section~\ref{sec-summary}, given $\theta$ and
$\epsilon$, we wish to find $k\geq 0$ and $u,t\in\Z[\omega]$ such that
\begin{equation}\label{eqn-def-u}
  U = \rtt{k}\zmatrix{cc}{u & -t\da \\ t & u\da}
\end{equation}
is unitary and satisfies $\norm{U-R_z(\theta)}\leq \epsilon$. We now elaborate how to
determine $k$, $u$, and $t$.

\subsection[The epsilon-region]{The $\epsilon$-region}\label{sec-epsilon-region}

We first note that $U$ is unitary if and only if $u\da u+t\da
t=2^k$. We now examine how to express the error $\epsilon$ as a function of
$u$; this will make explicit the set of available $u$ for a given
$\epsilon$. Let $z=x+yi = e^{-i\theta/2}$. First note that, since both
$U$ and $R_z(\theta)$ are $2\times 2$ unitary of determinant 1, the
operator norm of $U-R_z(\theta)$ coincides with $1/\sqrt2$ of the
Hilbert-Schmidt norm. It can therefore be calculated as
follows. Let $\uhat=\rtt{k}u=\alpha+\beta i$ and
$\that=\rtt{k}t$.
\begin{equation}
  \norm{U-R_z(\theta)}^2 = \frac12 \norm{U-R_z(\theta)}_{\HS}^2
  = |\uhat-z|^2 + |\that|^2.
\end{equation}
Using ${\uhat}\da \uhat + {\that}\da \that=1$ and $z\da z=1$, we can further simplify
this to:
\begin{equation}
  |\uhat-z|^2 + |\that|^2 = (\uhat-z)\da(\uhat-z) + {\that}\da \that
  = {\uhat}\da \uhat - {\uhat}\da z - z\da \uhat + z\da z + {\that}\da \that
  = 2 - 2\Realpart({\uhat}\da z)
  = 2 - 2\zmatrix{c}{\alpha\\\beta}\cdot\zmatrix{c}{x\\y}.
\end{equation}
The error is therefore directly related to the dot product of $\uhat$ and
$z$, regarded as vectors in $\R^2$. Writing $\vec z = (x,y)^T$ and
$\vec u = (\alpha,\beta)^T$, we have
\begin{equation}
  \norm{U-R_z(\theta)}^2 \leq \epsilon ^2
  \iff 2 - 2\vec u\cdot \vec z \leq \epsilon^2
  \iff \vec u\cdot \vec z \geq 1 - \frac{\epsilon^2}{2}.
\end{equation}
Let us define the {\em $\epsilon$-region} $\Repsilon$ to be the
corresponding subset of the unit disk. Let $\Disk=\s{\vec u ~\mid~
  |\vec u|^2 \leq 1}$ be the closed unit disk. Then
\begin{equation}
  \Repsilon =
  \s{\vec u\in\Disk \mid \vec u\cdot \vec z\geq 1 - \frac{\epsilon^2}{2}}.
\end{equation}
The $\epsilon$-region is shown as a shaded region in the illustration
below. Note that the width of this region is $\frac{\epsilon^2}{2}$ at
its widest point, and its length, to second order, is $2\epsilon$, and
in any case greater than $\sqrt2\,\epsilon$.
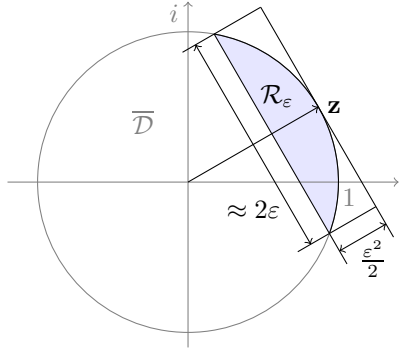
\[
\begin{tikzpicture}[scale=2]
  \fill[fill=blue!10, rotate=30] (cos 50, sin 50) -- (cos 50, -sin 50) arc (-50:50:1) -- cycle;
  \path[color=gray] (-.3,.4) node {$\Disk$};
  \draw[color=gray,->] (-1.2,0) -- (1.4,0);
  \draw[color=gray,->] (0,-1.1) -- (0,1.2);
  \path[color=gray] (0,1) node[above left] {$i$};
  \path[color=gray] (1,0) node[below=6pt,right=-2pt] {$1$};
  \draw[color=gray] (0,0) circle (1);
  \draw[rotate=30] (cos 50, -sin 50) arc (-50:50:1);
  \draw[->, rotate=30] (0,0) -- (1,0) node[right] {$\vec z$};
  \draw[rotate=30] (cos 50, sin 50) -- (cos 50,-1);
  \draw[rotate=30] (1,sin 50) -- (1,-1);
  \draw[rotate=30] (0.4, sin 50) -- (1, sin 50);
  \draw[rotate=30] (0.4, -sin 50) -- (1, -sin 50);
  \draw[<->, rotate=30] (cos 50,-0.9) -- (1,-0.9);
  \draw[rotate=30] (0.9, -0.9) node[below] {$\frac{\epsilon^2}{2}$};
  \draw[<->, rotate=30] (0.5, sin 50) -- (0.5, -sin 50) node[above=13pt,left=8pt] {$\approx 2\epsilon$};
  \path[rotate=30] (0.8,0.2) node {$\Repsilon$};
\end{tikzpicture}
\]
In summary, $\norm{U-R_z(\theta)}\leq \epsilon$ if and only if $\vec
u\in \Repsilon$. 

From now on, it will be convenient to identify the complex plane with
$\R^2$, i.e., we will regard $\Disk$ and $\Repsilon$ as subsets of the
complex numbers as well as $\R^2$.

\subsection{Candidates}

Assume, for the moment, that some suitable $k\geq 1$ has been fixed.
The problem of finding an $\epsilon$-approximation of $R_z(\theta)$
can now be reduced to the following 2-step problem:
\begin{enumerate}\alphalabels
\item[(a)] find $u\in\Z[\omega]$ such that $\uhat=u/{\rt{k}}$ is in the
  $\epsilon$-region, and
\item[(b)] find $t\in\Z[\omega]$ such that $t\da t+u\da u=2^k$.
\end{enumerate}
Recall from Section~\ref{sec-dioph} that we have an efficient solution
to (b) if $\xi=2^k-u\da u$ satisfies the hypotheses of
Theorem~\ref{thm-dioph}.  We will quite literally leave one of these
hypotheses, namely the primality of $\xi\bul\xi$, to chance. However, we
will arrange things such that the remaining hypotheses are satisfied
by design. To that end, we will say that $\uhat$ is a {\em candidate} if
it satisfies all the required conditions, except possibly for the
primality of $\xi\bul\xi$. This is made precise in the following
definition and lemma.

\begin{definition}
  Let $\epsilon>0$, $\theta\in\R$, and $k\geq 1$ be fixed. Let
  $\uhat=u/{\rt{k}}$, where $u\in\Z[\omega]$. Let
  $\xi=x+y\sqrt2=2^k-u\da u$. Then $\uhat$ is called a {\em candidate} if
  $\uhat\in \Repsilon$, $x$ is odd, $y$ is even, $\xi\geq 0$, and
  $\xi\bul\geq 0$.  We say that $\uhat$ is a {\em prime candidate} if,
  moreover, $p=\xi\bul\xi$ is prime.
\end{definition}

We will further restrict attention to candidates $\uhat=u/\rt{k}$
where $u$ is from the ring $\Z[\sqrt2,i]\seq\Z[\omega]$, i.e., of the
form $u=\alpha + \beta i$ with $\alpha,\beta\in\Z[\sqrt2]$.

\begin{lemma}\label{lem-candidate}
  Let $\uhat=u/\rt{k}$, where $u=\alpha+\beta i$ and
  $\alpha,\beta\in\Z[\sqrt2]$. Let us write $\alpha=a+b\sqrt2$ and
  $\beta=c+d\sqrt2$.  Then $\uhat$ is a candidate if and only if the
  following three conditions hold:
  \begin{itemize}
  \item $\uhat\in \Repsilon$,
  \item ${\uhat}\bul \in \Disk$, and
  \item $a+c$ is odd.
  \end{itemize}
\end{lemma}

\begin{proof}
  With the given choice of notation, we have
  \begin{equation}
    \begin{split}
      \xi &= 2^k-u\da u
      \\&= 2^k - \alpha^2 - \beta^2
      \\&= (2^k - a^2-2b^2 -c^2-2d^2)- (2ab+2cd)\sqrt2,
    \end{split}
  \end{equation}
  hence $x = 2^k - a^2-2b^2 -c^2-2d^2$ and $y = 2(ab+cd)$. Then $x$ is
  odd if and only if $a+c$ is odd. The condition that $y$ is even is
  automatically satisfied. The condition $\xi\geq 0$ is equivalent to
  $u\da u \leq 2^k$, or equivalently ${\uhat}\da \uhat\leq 1$, i.e., $\uhat\in
  \Disk$, which follows from $\uhat\in \Repsilon$. Similarly, $\xi\bul\geq 0$ is
  equivalent to ${\uhat}\bul\in\Disk$.
\end{proof}

\subsection{Candidate selection}

\begin{theorem}\label{thm-candidate-selection}
  Let $\epsilon>0$ and $\theta$ be fixed, and let $k\geq
  C+2\log_2(1/\epsilon)$, where $C=\frac{5}{2}+2\log_2(1+\sqrt2)$.  Then there
  exists a set of at least $n=\floor{\frac{4\sqrt2}{\epsilon}}$
  candidates $\uhat$ satisfying the conditions of
  Lemma~\ref{lem-candidate}. Moreover, there is an efficient algorithm
  for computing a random candidate from this set.
\end{theorem}

\begin{proof}
  First note that $k\geq C+2\log_2(1/\epsilon)$ implies $2^k\geq
  \frac{4\sqrt2(1+\sqrt2)^2}{\epsilon^2}$. Define
  \begin{equation}
    \delta = \rt{k}\,\frac{\epsilon^2}{8}
    \quad\mbox{and}\quad
    \Delta=\rt{k+1}.
  \end{equation}
  We note that $\delta$ and $\Delta$ satisfy the condition of
  Lemma~\ref{lem-interval}. Indeed,
  \begin{equation}\label{eqn-delta-new}
    \delta\Delta =
    2^k\,\frac{\sqrt2\,\epsilon^2}{8}
    \geq \frac{8(1+\sqrt2)^2}{\epsilon^2}\,\frac{\epsilon^2}{8}
    = (1+\sqrt2)^2.
  \end{equation}
  In the following, we will assume, for convenience, that
  $-\frac{\pi}{2}\leq\theta\leq\frac{\pi}{2}$. This assumption is
  without loss of generality, because otherwise, we may simply rotate
  the $\epsilon$-region by multiples of $90^{\circ}$ without changing
  the substance of the argument.

  Now consider the line $\vec u\cdot\vec z = 1 - \epsilon^2/4$. It
  intersects the unit circle in two points at $y$-coordinates $\ymin$
  and $\ymax$, with $\ymax-\ymin\geq\frac{\epsilon}{\sqrt2}$, as shown in
  this figure:
  
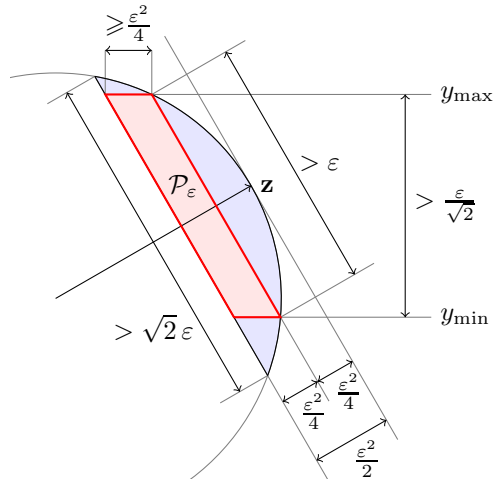
\[
  \begin{tikzpicture}[scale=3]
    \clip (-0.2,-0.8) rectangle (2.4,1.3);
    \fill[fill=blue!10, rotate=30] (cos 50, sin 50) -- (cos 50, -sin 50) arc (-50:50:1) -- cycle;
    \fill[color=red!10,rotate=30] (0.8213938,-0.57036148) --
    (0.8213938,0.57036148) -- (cos 50,0.57036148+0.10311833) -- (cos 50,
    -0.57036148+0.10311833) -- cycle;
    \draw[color=gray,rotate=30] (0.8213938, -0.57036148) -- (1.3, -0.57036148);
    \draw[color=gray,rotate=30] (0.8213938, 0.57036148) -- (1.3, 0.57036148);
    \path[color=gray] (-.3,.4) node {$\Disk$};
    \draw[color=gray] (0,0) circle (1);
    \draw[color=gray,rotate=30] (0.4, sin 50) -- (cos 50, sin 50);
    \draw[color=gray,rotate=30] (0.4, -sin 50) -- (cos 50, -sin 50);
    \draw[color=gray,rotate=30] (0.8213938, -0.57036148) -- (0.8213938,-1.0);
    \draw[color=gray,rotate=30] (1,sin 50) -- (1,-1.3);
    \draw[color=gray,rotate=30] (cos 50, sin 50) -- (cos 50, -1.3);
    \draw[draw=gray,rotate=30] (0.8213938, -0.57036148) --
    (1.4,-0.57036148-1.4*tan 30+0.8213938*tan 30) node[right]{$\ymin$};
    \draw[draw=gray,rotate=30] (0.8213938, 0.57036148) --
    (1.4+2*0.57036148*sin 30*cos 30,0.57036148-1.4*tan 30-2*0.57036148*sin
    30*cos 30*tan 30+0.8213938*tan 30) node[right]{$\ymax$};
    \draw[draw=gray] (0.4261671,0.9046444) -- (0.4261671,1.1);
    \draw[draw=gray] (0.2199304,0.9046444) -- (0.2199304,1.1);
    \draw[<->] (0.4261671,1.1) -- node[above] {${\geq}\frac{\epsilon^2}{4}$} (0.2199304,1.1);
    \draw[<->,rotate=30]
    (1.3+2*0.57036148*sin 30*cos 30,0.57036148-1.3*tan 30-2*0.57036148*sin
    30*cos 30*tan 30+0.8213938*tan 30) -- node[right] {$> \frac{\epsilon}{\sqrt2}$}
    (1.3,-0.57036148-1.3*tan 30+0.8213938*tan 30);
    \draw[rotate=30] (cos 50, -sin 50) arc (-50:50:1);
    \draw[->, rotate=30] (0,0) -- (1,0) node[right] {$\vec z$};
    \draw[rotate=30] (cos 50, sin 50) -- (cos 50, -sin 50);
    \draw[<->, rotate=30] (cos 50,-1.2) -- (1,-1.2);
    \draw[rotate=30] (0.9, -1.2) node[below] {$\frac{\epsilon^2}{2}$};
    \draw[<->, rotate=30] (cos 50,-0.9) -- (0.8213938,-0.9);
    \draw[rotate=30] (0.8, -0.9) node[below] {$\frac{\epsilon^2}{4}$};
    \draw[<->, rotate=30] (0.8213938,-0.9) -- (1,-0.9);
    \draw[rotate=30] (0.98, -0.9) node[below] {$\frac{\epsilon^2}{4}$};
    \draw[<->, rotate=30] (0.5, sin 50) -- (0.5, -sin 50)
    node[above=24pt,left=14pt] {$> \sqrt2\,\epsilon$};
    \draw[color=red,thick,rotate=30] (cos 50, -0.57036148+0.10311833) --
    (0.8213938,-0.57036148) -- (0.8213938,0.57036148) -- (cos
    50,0.57036148+0.10311833) -- cycle;
    \draw[<->,rotate=30] (1.2, -0.57036148) -- node[right] {$>\epsilon$} (1.2, 0.57036148);
    \path[rotate=30] (0.74,0.15) node {$\Pepsilon$};
  \end{tikzpicture}
  \]
  Consider the parallelogram $\Pepsilon$ that is bounded by the lines
  $\vec u\cdot\vec z = 1 - \epsilon^2/4$ and $\vec u\cdot\vec z = 1 -
  \epsilon^2/2$, and by the horizontal lines at $y=\ymin$ and
  $y=\ymax$. As illustrated in the above figure, $\Pepsilon$ is
  entirely contained within the $\epsilon$-region. We will select
  candidates from within $\Pepsilon$.

  Let $n=\floor{\frac{4\sqrt2}{\epsilon}}$, and define $y_j = \ymin +
  \frac{j}{n}(\ymax-\ymin)$ for $j=0,\ldots,n$, so that $\ymin = y_0 <
  y_1 < \ldots < y_n = \ymax$. We note that
  \begin{equation}
    y_{j+1} - y_j = \frac{\ymax-\ymin}{n} >
    \frac{\epsilon}{\sqrt2}\,\frac{\epsilon}{4\sqrt2}
    = \frac{\epsilon^2}{8}.
  \end{equation}
  Let $I_j=[y_j,y_j+\frac{\epsilon^2}{8}]$ for $j=0,\ldots,n-1$. Then
  $I_0,\ldots,I_{n-1}$ are non-overlapping closed subintervals of
  $[\ymin,\ymax]$ of size $\frac{\epsilon^2}{8}$.  For each
  $j=0,\ldots,n-1$, we will find a candidate $\uhat\in \Pepsilon\cap
  (\R\times I_j)$ as follows.

  First, use Lemma~\ref{lem-interval} to find $\beta\in\Z[\sqrt2]$
  such that $\beta\in[\rt{k} y_j,\rt{k}
  (y_j+\frac{\epsilon^2}{8})]$ and
  $\beta\bul\in[-\rt{k-1},\rt{k-1}]$. Note that these
  two intervals are of size $\delta$ and $\Delta$, respectively, so
  the use of Lemma~\ref{lem-interval} is justified by
  (\ref{eqn-delta-new}).  Let $\betahat = \beta / \rt{k} \in I_j$.

  Because $\betahat\in[\ymin,\ymax]$, the line $y=\betahat$ intersects the
  parallelogram $\Pepsilon$, as shown in the figure below. Let
  $x_0=\min\s{x\mid (x,\betahat)\in \Pepsilon}$ and $x_1=\max\s{x\mid
    (x,\betahat)\in \Pepsilon}$, and note that
  $x_1-x_0\geq\frac{\epsilon^2}{4}$.
  
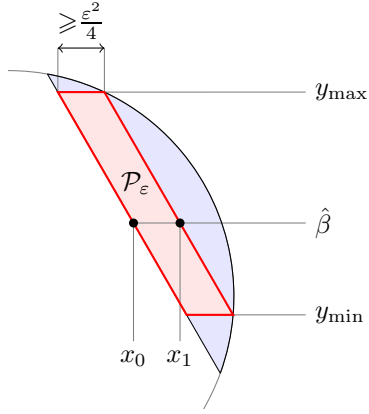
\[
  \begin{tikzpicture}[scale=3]
    \clip (-0,-0.5) rectangle (2,1.3);
    \fill[fill=blue!10, rotate=30] (cos 50, sin 50) -- (cos 50, -sin 50) arc (-50:50:1) -- cycle;
    \fill[color=red!10,rotate=30] (0.8213938,-0.57036148) --
    (0.8213938,0.57036148) -- (cos 50,0.57036148+0.10311833) -- (cos 50,
    -0.57036148+0.10311833) -- cycle;
    \draw[color=gray] (0,0) circle (1);
    \draw[draw=gray,rotate=30] (0.8213938, -0.57036148) --
    (1.1,-0.57036148-1.1*tan 30+0.8213938*tan 30) node[right]{$\ymin$};
    \draw[draw=gray,rotate=30] (0.8213938, 0.57036148) --
    (1.1+2*0.57036148*sin 30*cos 30,0.57036148-1.1*tan 30-2*0.57036148*sin
    30*cos 30*tan 30+0.8213938*tan 30) node[right]{$\ymax$};
    \draw[draw=gray,rotate=30] (cos 50, -0.1+0.10311833) --
    (1.1-0.1*sin 30*cos 30+0.57036148*sin 30*cos 30,-0.1-1.1*tan 30+0.1*sin
    30*cos 30*tan 30-0.57036148*sin
    30*cos 30*tan 30+0.8213938*tan 30) node[right]{$\betahat$};
    \draw[draw=gray] (0.5551112,0.3240943) -- (0.5551112,-0.2) node[below]{$x_0$};
    \draw[draw=gray] (0.7613478,0.3240943) -- (0.7613478,-0.2) node[below]{$x_1$};
    \draw[draw=gray] (0.4261671,0.9046444) -- (0.4261671,1.1);
    \draw[draw=gray] (0.2199304,0.9046444) -- (0.2199304,1.1);
    \draw[<->] (0.4261671,1.1) -- node[above] {${\geq}\frac{\epsilon^2}{4}$} (0.2199304,1.1);
    \draw[rotate=30] (cos 50, -sin 50) arc (-50:50:1);
    \draw[rotate=30] (cos 50, sin 50) -- (cos 50, -sin 50);
    \draw[color=red,thick,rotate=30] (cos 50, -0.57036148+0.10311833) --
    (0.8213938,-0.57036148) -- (0.8213938,0.57036148) -- (cos
    50,0.57036148+0.10311833) -- cycle;
    \path[rotate=30] (0.74,0.15) node {$\Pepsilon$};
    \fill (0.5551112,0.3240943) circle (0.02);
    \fill (0.7613478,0.3240943) circle (0.02);
  \end{tikzpicture}
  \]
  Now we can use Corollary~\ref{cor-interval-even-odd} to find
  $\alpha\in\Z[\sqrt2]$ such that
  $\alpha\in[\rt{k}x_0,\rt{k}(x_0+\frac{\epsilon^2}{4})]$
  and $\alpha\bul\in[-\rt{k-1},\rt{k-1}]$. Note that
  these two intervals are of size $2\delta$ and $\Delta$, respectively,
  so that the use of Corollary~\ref{cor-interval-even-odd} is again
  justified by (\ref{eqn-delta-new}). Furthermore, writing
  $\alpha=a+b\sqrt2$ and $\beta=c+d\sqrt2$,
  Corollary~\ref{cor-interval-even-odd} permits us to choose $a$ odd
  if $c$ is even, or vice versa.

  Let $\alphahat=\alpha/\rt{k}$, and let $\uhat=\alphahat+\betahat i$.  It
  is now trivial to show that $\uhat$ is a candidate. Indeed, since
  $\alphahat\in[x_0,x_1]$, we have that $\uhat\in \Pepsilon\seq
  \Repsilon$ by construction. Also, note that, by construction,
  $({\alphahat}\bul, {\betahat}\bul) \in
  [-\frac{1}{\sqrt2},\frac{1}{\sqrt2}]^2 \seq \Disk$, so that
  ${\uhat}\bul\in\Disk$. Finally, we already remarked that $a+c$ is odd.

  Since we have found a distinct candidate for each
  $j\in\s{0,\ldots,n-1}$, there exist at least $n$ candidates;
  moreover, the construction of the $j$th candidate is clearly algorithmic.
\end{proof}

\section{The main algorithm}

\subsection{Approximating a $z$-rotation}

We are now ready to put together the results of earlier sections to
obtain an algorithm for approximating $R_z(\theta)$ up to $\epsilon$,
using only gates from the single-qubit Clifford+$T$ group.

\begin{algorithm}\label{alg-rotation}
  Inputs: $0<\epsilon\leq\frac12$ and $\theta\in\R$. Let
  $k=\ceil{C+2\log_2(1/\epsilon)}$, where
  $C=\frac{5}{2}+2\log_2(1+\sqrt2)\approx 5.04$, and let
  $n=\floor{\frac{4\sqrt2}{\epsilon}}$. Use
  Theorem~\ref{thm-candidate-selection} to generate random candidates
  $\uhat\in\rtt{k}\Z[\omega]$. For each candidate
  $\uhat=u/\rt{k}$, let $\xi=2^k-u\da u\in\Z[\sqrt2]$ and attempt to
  use the method of Theorem~\ref{thm-dioph} to find $t\in\Z[\omega]$
  with $t\da t = \xi$. If this fails, continue with the next candidate. If it
  succeeds, the operator
  \[   U = \rtt{k}\zmatrix{cc}{u & -t\da \\ t & u\da}
  \]
  is the desired $\epsilon$-approximation of $R_z(\theta)$. Use the
  Kliuchnikov-Maslov-Mosca exact synthesis algorithm
  {\cite{Kliuchnikov-etal}} to convert $U$ to a sequence of
  Clifford+$T$ gates; optionally, reduce the sequence of gates to
  Matsumoto-Amano normal form {\cite{MA08}}. Output: the sequence of
  gates.
\end{algorithm}

\begin{remark}
  It is not necessary to limit the generation of candidates to the
  particular set of $n$ candidates identified in
  Theorem~\ref{thm-candidate-selection}. Instead, one can simply
  generate candidates at random until a solution is found.
\end{remark}

\begin{remark}
  Although the method of Theorem~\ref{thm-dioph} will usually only
  work when $p=\xi\bul\xi$ is prime, Algorithm~\ref{alg-rotation}
  requires no explicit primality test. Instead, one can simply perform
  the method of Theorem~\ref{thm-dioph} under the optimistic assumption
  that $p$ is prime. In the worst case, this may yield a $t$ that is
  not a solution, but this can be easily checked after the fact. Some
  slight care is needed to ensure that the probabilistic algorithm
  from Remark~\ref{rem-h}, for finding a square root $h$ of $-1$
  modulo $p$, does not get into an infinite loop when $p$ is not
  prime. For this, it is sufficient to limit the number of iterations
  of this algorithm to some small number, say 1 or 2. If $p$ is indeed
  prime, this will still yield $h$ with high enough probability; in
  the worst case, some prime candidates will be unnecessarily
  rejected.
\end{remark}

\subsection{Approximating arbitrary gates}

To approximate an arbitrary gate $U\in\SU(2)$, first decompose it via
Euler angles as $U=R_z(\beta)R_x(\gamma)R_z(\delta) =
R_z(\beta)\,H\,R_z(\gamma)\,H\,R_z(\delta)$; each rotation can then be
approximated separately, say up to $\epsilon/3$.

\begin{remark}
  If one includes, as we did, only global phases of the form
  $\omega^j$ in the Clifford group, it is of course not possible to
  approximate unitary operators of arbitrary determinant, using
  only single-qubit Clifford+$T$ gates.  Since the determinant of
  every Clifford+$T$ operator is a power of $\omega$, an operator
  $U\in\U(2)$ can be approximated to arbitrary $\epsilon$ if and only
  if $\det U$ is also a power of $\omega$. Otherwise, $U$ can only be
  approximated up to a global phase. If we include arbitrary global
  phases in the Clifford group, then of course all operators can be
  approximated.
\end{remark}

\section{Complexity analysis}\label{sec-complexity}

\subsection{Gate complexity}

\begin{fact}\label{fact-2k}
  Let $U$ be a single-qubit Clifford+$T$ operator with denominator
  exponent $k>0$, and let $n$ be the minimal $T$-count of $U$. Then
  $2k-3\leq n\leq 2k$.
\end{fact}

\begin{proof}
  By induction on the Matsumoto-Amano normal form of $U$; see
  Theorem~7.10 of {\cite{ma-remarks}}.
\end{proof}

Since the matrix $U$ constructed by Algorithm~\ref{alg-rotation} has
denominator exponent $k$, it follows by Fact~\ref{fact-2k} that the
approximation of a $z$-rotation uses $T$-count at most $2k$, or
$K+4\log_2(1/\epsilon)$, where $K\approx 10.09$. The approximation of
an arbitrary element of $\SU(2)$ requires three $z$-rotations, so the
$T$-count is $K+12\log_2(1/\epsilon)$, where $K\approx 30.26$.

\begin{remark}
  Rotations about other ``easy'' axes, such as $x$-rotations or
  $y$-rotations, can of course also be approximated with the same
  $T$-count as a $z$-rotation, as they differ from a $z$-rotation only
  by Clifford operators. More generally, rotations of the form
  $VR_z(\theta)V\da$, where $V$ is a fixed Clifford+$T$-operator, can
  be approximated with $T$-count $K+4\log_2(1/\epsilon)$, where $K$ is
  a constant depending only on $V$.
\end{remark}

\void{
\begin{remark}
  We have used $T$-count as a convenient measure for circuit size for
  three reasons. First, the fault-tolerant implementation of $T$-gates
  is far more costly than that of Clifford gates, so that $T$-count
  tends to dominate the overall implementation cost. Second, in counting
  $T$-gates, we have skirted the question for how exactly Clifford gates
  should be counted, e.g., whether consecutive Clifford gates should be
  regarded as a single gate. Third, by defining $T_z=T$, $T_x=HTH$, and
  $T_y=SHTHS\da$, every Clifford+$T$ circuit can be written in the form
  \begin{equation}
    T_1T_2\ldots T_nC,
  \end{equation}
  where $T_1,\ldots,T_n\in\s{T_x,T_y,T_z}$. Since $T_x$ and $T_y$ are
  Clifford-conjugates of $T$, they are presumably no easier or harder
  to perform; in other words, except for a single Clifford operator at
  the end, all gates are essentially $T$-gates.
\end{remark}
}

\subsection{Time complexity}

Most parts of the algorithm are computationally straightforward; for
example, the generation of candidates, and the various ring operations
required for Theorem~\ref{thm-dioph}, can all be done with integers of
length $O(k)$. The arithmetic operations, such as addition,
multiplication, division, etc., require no more than $O(k^2)$
elementary steps each, and there are $O(k)$ arithmetic operations to
perform.

The dominant complexity question is how many candidates must be
generated before a solution is found, and indeed, whether a solution
will be found at all. Here, we must make an assumption about the
distribution of primes:

\begin{hypothesis}\label{hyp-random}
  For a randomly chosen candidate, the probability that $p=\xi\bul\xi$
  is prime is asymptotically no smaller than for general odd numbers
  of comparable size.
\end{hypothesis}

\noindent
We note that for a candidate $\uhat=u/\rt{k}$, we have
\begin{equation}
  {\uhat}\da \uhat\geq \left(1-\frac{\epsilon^2}{2}\right)^2 \geq 1-\epsilon^2,
\end{equation}
and thus
\begin{equation}
  \xi ~=~ 2^k-u\da u ~=~ 2^k(1-{\uhat}\da \uhat)
  ~\leq~ 2^k\epsilon^2 ~\leq~
  \frac{4\sqrt2(1+\sqrt2)^2}{\epsilon^2}\,\epsilon^2
  = 4\sqrt2(1+\sqrt2)^2 \leq 33.  
\end{equation}
Also, $\xi\bul = 2^k-u\bulda u\bul \leq 2^k$. It follows that
$p=\xi\bul\xi \leq 33\cdot 2^k$. By the prime number theorem, a
randomly chosen odd $p$ in this range has probability
\begin{equation}\label{eqn-probability}
  P\approx\frac{2}{\ln (33\cdot 2^k)} = \Omega(\frac{1}{\,k\,})
\end{equation}
of being prime; by Hypothesis~\ref{hyp-random}, at least the same
probability holds for a randomly chosen candidate. On the other hand,
there are $n=\floor{\frac{4\sqrt2}{\epsilon}} = O(\rt{k})$
candidates available, so asymptotically, a prime will be found with
certainty. The expected number of top-level iterations of
Algorithm~\ref{alg-rotation} until a solution is found is $O(k)$.
Therefore, under Hypothesis~\ref{hyp-random}, the overall runtime of
the algorithm is no more than $O(k^4) = O(\log^4(1/\epsilon))$. A
sharper bound could probably be derived by a more sophisticated analysis.

We also note that, due to its randomized nature, and because each
candidate is chosen independently, the algorithm can easily be
parallelized. 

\subsection{Seeding}

While the above discussion concerns the asymptotic case, one may
wonder whether for {\em particular} values of $\epsilon$ and $\theta$
it may happen, by unlucky coincidence, that none of the available
candidates are prime. In practice, this never seems to be the case, as
primes are always found quite quickly. However, in theory, we may
avoid this situation by the method of {\em seeding}: Instead of
approximating $R_z(\theta)$ directly, we may choose a random angle
(``seed'') $\phi$, approximate both $R_z(\phi)$ and $R_z(\theta-\phi)$
to within $\epsilon/2$, and finally compute $R_z(\theta)$ as
$R_z(\phi)R_z(\theta-\phi)$. If a particular seed seems to yield no
prime candidates in a given number of iterations, one can simply try a
different seed.

The seeding method for $z$-rotations has the disadvantage that it
approximately doubles the $T$-count, from $K+4\log_2(1/\epsilon)$ to
$K+8\log_2(1/\epsilon)$.

In the approximation of an arbitrary gate $U\in\SU(2)$, a different
seeding method can be used, which only increases the asymptotic
$T$-count by a small additive constant. Namely, let the seed be some
randomly chosen Clifford+$T$ operator $V$ of fixed and relatively
small $T$-count.  We can then use the algorithm to approximate
$UV\inv$, and multiply the final result by $V$. The gate complexity,
in this case, is unchanged at $K+12\log_2(1/\epsilon)$.

It must be emphasized, once again, that seeding is only of theoretical
interest, to ensure expected termination of the algorithm under a
hypothetical worst-case scenario. In practice, it does not seem to be
required.

\section{Lower bounds}\label{sec-lower-bounds}

As mentioned in the introduction, $K+3\log_2(1/\epsilon)$ is an easy
lower bound for the $T$-count of a Clifford+$T$ approximation of some
arbitrary operator in $\SU(2)$ up to $\epsilon$.
Specifically, Matsumoto and Amano showed that there are precisely
$192\cdot(3\cdot 2^n-2)$ distinct single-qubit Clifford+$T$-circuits
of $T$-count at most $n$ {\cite{MA08}}. Since $\SU(2)$ is a
3-dimensional manifold, it requires $\Omega(1/\epsilon^3)$
epsilon-balls to cover. The resulting inequality
\begin{equation}\label{eqn-ma-bound}
  192\cdot(3\cdot 2^n-2)\geq \frac{c}{\epsilon^3},
\end{equation}
immediately implies 
\begin{equation}\label{eqn-ma-bound2}
  n\geq K+3\log_2(\frac{1}{\,\epsilon\,}).
\end{equation}
This means that there is some universal constant $K$ such that for
every $\epsilon$, there exists some operator $U\in\SU(2)$ that cannot
be approximated up to $\epsilon$ with $T$-count less than
$K+3\log_2(1/\epsilon)$. In fact, (\ref{eqn-ma-bound}) implies that
the proportion (in the Haar measure, say) of operators that can be
approximated with $T$-count $K+3\log_2(1/\epsilon)-k$ is at most
$O(\frac{1}{2^k})$, so that ``most'' operators require $T$-counts that
are close to or above the lower bound.

However, like all lower bounds, this must be understood with a grain
of salt. For example, the set of $z$-rotations only forms a
$1$-dimensional submanifold of $\SU(2)$, so it is not a priori clear
that the lower bound of $K+3\log_2(1/\epsilon)$ also applies to
$z$-rotations.

We now prove a sharper lower bound, which applies to $z$-rotations in
particular. 

\begin{theorem}\label{thm-lower-bound}
  Let $K'=-9$. Then for all $\epsilon>0$, there exists an angle
  $\theta$ such that every $\epsilon$-approximation of $R_z(\theta)$
  by a product of Clifford+$T$ operators requires $T$-count at least
  $K'+4\log_2(1/\epsilon)$.
\end{theorem}

\begin{proof}
  Let $\epsilon$ be given. If $\epsilon\geq 1/2$, then
  $K'+4\log_2(1/\epsilon)$ is negative, so there is nothing to
  show. Assume, therefore, that $\epsilon<1/2$.  Let
  $\phi=\sin\inv\epsilon$, with $0<\phi<\pi/6$. Let $\theta=-2\phi$,
  so that $z = e^{-i\theta/2} = e^{i\phi}$. We note the shape of the
  $\epsilon$-region for $z$:
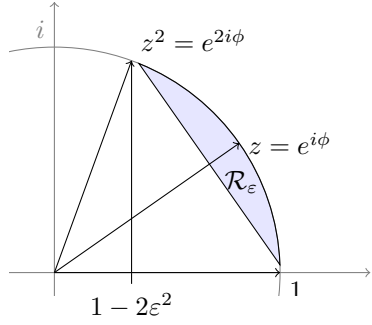
  \[
  \begin{tikzpicture}[scale=3]
    \draw (cos 70, sin 70) -- (cos 70, -.05) node [below] {$1-2\epsilon^2$};
    \clip (-0.2,-0.1) rectangle (1.41,1.21);
    \fill[fill=blue!10] (cos 68, sin 68) -- (cos 2, sin 2) arc (2:68:1) -- cycle;
    \path[color=gray] (-.3,.4) node {$\Disk$};
    \draw[color=gray,->] (-1.2,0) -- (1.4,0);
    \draw[color=gray,->] (0,-1.1) -- (0,1.2);
    \path[color=gray] (0,1) node[above left] {$i$};
    \draw[color=gray] (0,0) circle (1);
    \draw (cos 68, sin 68) -- (cos 2, sin 2) arc (2:68:1) -- cycle;
    \draw[->] (0,0) -- (cos 35, sin 35) node[right] {$z = e^{i\phi}$};
    \draw[->] (0,0) -- (cos 0, sin 0) node[below right] {$1$};
    \draw[->] (0,0) -- (cos 70, sin 70) node[above right] {$z^2 = e^{2i\phi}$};
    \path (0.83,0.4) node {$\Repsilon$};
  \end{tikzpicture}
  \]
  Identifying complex numbers with vectors in $\R^2$ as before, we
  estimate the dot product 
  \begin{equation}
    \vec 1\cdot \vec z = \cos \phi = \sqrt{1-\epsilon^2}
    < \sqrt{1-\epsilon^2+\frac{\epsilon^4}{4}}
    = 1-\frac{\epsilon^2}{2}.
  \end{equation}
  It follows that neither $1$ nor $z^2$ is in the $\epsilon$-region.
  We further note that the $x$-coordinate of $z^2$ is $\cos 2\phi =
  1-2\sin^2\phi = 1-2\epsilon^2$. So for all $x+iy\in\Repsilon$, we
  have
  \begin{equation}\label{eqn-bound-reps}
    1-2\epsilon^2 < x < 1.
  \end{equation}
  Now consider some $\epsilon$-approximation 
  \begin{equation}\label{eqn-u-lower-bound}
    U = \rtt{k}\zmatrix{cc}{u & s \\ t & v}
  \end{equation}
  of $R_z(\theta)$ with denominator exponent $k$, i.e., where
  $u,t,s,v\in\Z[\omega]$. From $\norm{U-R_z(\theta)}\leq \epsilon$, we
  get $|\uhat-z|^2+|\that|^2\leq\epsilon^2$, where $\uhat=u/\rt{k}$
  and $\that=t/\rt{k}$. By the same reasoning as in
  Section~\ref{sec-epsilon-region}, it follows that
  $\uhat$ is in the $\epsilon$-region. Moreover since $U\bul$ is
  also unitary, we have ${\uhat}\bul\in\Disk$. We can write
  \[ \uhat=\rtt{k}(a\omega^3+b\omega^2+c\omega+d) = 
  \rtt{k+1}(d\sqrt2 + c-a) + (b\sqrt2 + c+a)i
  = \rtt{k+1}(\alpha + \beta i),
  \]
  where $\alpha,\beta\in\Z[\sqrt2]$.
  Since $\uhat\in\Repsilon$, from (\ref{eqn-bound-reps}), we have
  \begin{equation}
    \rt{k+1}(1-2\epsilon^2) < \alpha < \rt{k+1}.
  \end{equation}
  Also, from ${\uhat}\bul\in\Disk$, we have 
  \begin{equation}
    \alpha\bul \in [-\rt{k+1},\rt{k+1}].
  \end{equation}
  It follows that the constraints
  \begin{equation}
    (\gamma,\gamma\bul)\in[\rt{k+1}(1-2\epsilon^2),
    \rt{k+1}]\times [-\rt{k+1},\rt{k+1}]
  \end{equation}
  have at least two different solutions in $\Z[\sqrt2]$, namely
  $\gamma=\alpha$ and $\gamma=\rt{k+1}$. Setting
  $\delta=\rt{k+1}\,2\epsilon^2$ and
  $\Delta=2\rt{k+1}$, we have by 
  Lemma~\ref{lem-interval-lower}:
  \begin{equation}
    1 \leq \delta\Delta = 2^{k+3}\epsilon^2,
  \end{equation}
  or equivalently,
  \begin{equation}\label{eqn-lowerbound-k}
    k \geq \log_2(\frac{1}{8\epsilon^2})  = -3 +2\log_2(\frac{1}{\,\epsilon\,}).
  \end{equation}
  On the other hand, it follows from Fact~\ref{fact-2k} that every
  single-qubit Clifford+$T$ operator $U$ of $T$-count $n$ can be
  written with denominator exponent $k$, where
  \begin{equation}\label{eqn-lowerbound-n}
    k \leq \frac{n+3}{2}.
  \end{equation}
  Putting together (\ref{eqn-lowerbound-k}) and
  (\ref{eqn-lowerbound-n}), we get
  \begin{equation}\label{eqn-lowerbound}
    n \geq -9 + 4\log_2(\frac{1}{\,\epsilon\,}),
  \end{equation}
  which is the desired result.
\end{proof}

\begin{remark}
  Unlike the lower bound (\ref{eqn-ma-bound2}), which applies to {\em
    typical} operators, the lower bound (\ref{eqn-lowerbound}) only
  applies to carefully chosen worst-case operators. It is plausible
  that for any fixed $\epsilon$, approximations of order
  (\ref{eqn-ma-bound2}) can be achieved for most angles
  $\theta$. Moreover, it is also plausible that for any fixed
  $\theta$, approximations of order (\ref{eqn-ma-bound2}) can be
  achieved as $\epsilon\rightarrow 0$. A variation of
  Algorithm~\ref{alg-rotation} that could potentially achieve this is
  sketched in Section~\ref{sec-overclocking}, but the details are
  left for future work.
\end{remark}

\section{Overclocking}\label{sec-overclocking}

It is in the nature of Algorithm~\ref{alg-rotation} that $k$ is chosen
at the very beginning; the final sequence of gates will have $T$-count
very close to $2k$. This behavior is different from that of
search-based algorithms, which would typically try shorter
decompositions first, and then gradually move to longer ones.

In practice, many of the approximations and estimates in the above
proofs are conservative; it is often possible to choose a smaller $k$
than that prescribed by the algorithm, and still obtain a
decomposition. We refer to this technique as ``overclocking'' the
algorithm, by analogy with the practice of running microprocessors at
higher clock speeds than they were designed for, and hoping for the
best.

Let us first consider the effect of overclocking by a small additive
constant, i.e., decreasing the value of $C$ in
Algorithm~\ref{alg-rotation} by some fixed amount, independently of
$\epsilon$. This decreases the width of the $\epsilon$-region by a
fixed multiplicative factor, which means that some choices of $\beta$
(in Theorem~\ref{thm-candidate-selection}) will no longer yield a
successful solution for $\alpha$. However, the {\em proportion} of
$\beta$ for which an $\alpha$ can be found will be asymptotically
independent of $\epsilon$, so that the algorithm will still yield a
solution, albeit with a longer running time. In summary, the
asymptotic effect of overclocking $C$ by a fixed additive amount is to
increase the runtime by a fixed factor, while neither affecting the
validity nor the big-O time complexity of the algorithm.

A more interesting question is whether the algorithm can be modified
to allow overclocking by a {\em multiplicative} factor; in the ideal
case, one might even hope to achieve approximations with $T$-count
$K+3\log_2(1/\epsilon)$, for most angles $\theta$. As it is currently
stated, the algorithm will not work well for multiplicative
overclocking, because the $\epsilon$-region will then be much too thin
to expect to find candidates for a reasonable fraction of
$y$-coordinates $\beta$. However, the {\em area} of the
$\epsilon$-region scales as $\epsilon^3$, and since the density of
grid points scales as $4^k \sim 1/\epsilon^3$, one still expects, on
average, to find a constant number of candidates in the
$\epsilon$-region. This assumes that the angle $\theta$ is
sufficiently general, so that the $\epsilon$-region is not aligned
with the orientation of the underlying $k$-grid, so that the lower
bound of Theorem~\ref{thm-lower-bound} can be avoided.  It would be an
interesting optimization to refine Algorithm~\ref{alg-rotation} to be
able to take advantage of such sparsely populated $\epsilon$-regions,
but the details are left to future work.

\section{Experimental results}

\begin{table}
\[
\begin{array}{|l|r|r|l|r|r|c|}
  \hline
  \multicolumn{1}{|c|}{\epsilon}
  & \multicolumn{1}{|c|}{\mbox{$k$}}
  & \multicolumn{1}{|c|}{\mbox{$T$-count}}
  & \multicolumn{1}{|c|}{\mbox{Actual error}}
  & \multicolumn{1}{|c|}{\mbox{Runtime}}
  & \multicolumn{1}{|c|}{\mbox{Candidates}}
  & \multicolumn{1}{|c|}{\mbox{Time/Candidate}}
  \\\hline
    10^{-10}
  & 72
  & 142
  & 0.90665\cdot 10^{-10}
  & 0.02s
  & 37.80
  & 0.0005s
  \\

  10^{-20}
  & 139
  & 278
  & 0.84346\cdot 10^{-20}
  & 0.05s
  & 86.60
  & 0.0005s
  \\

  10^{-30}
  & 205
  & 410
  & 0.82984\cdot 10^{-30}
  & 0.07s
  & 97.94
  & 0.0007s
  \\

  10^{-40}
  & 272
  & 542
  & 0.73279\cdot 10^{-40}
  & 0.10s
  & 114.00
  & 0.0009s
  \\

  10^{-50}
  & 338
  & 676
  & 0.83841\cdot 10^{-50}
  & 0.13s
  & 135.66
  & 0.0010s
  \\

  10^{-60}
  & 405
  & 810
  & 0.73964\cdot 10^{-60}
  & 0.21s
  & 213.22
  & 0.0010s
  \\

  10^{-70}
  & 471
  & 942
  & 0.72360\cdot 10^{-70}
  & 0.24s
  & 201.76
  & 0.0012s
  \\

  10^{-80}
  & 538
  & 1076
  & 0.96804\cdot 10^{-80}
  & 0.30s
  & 223.72
  & 0.0013s
  \\

  10^{-90}
  & 604
  & 1208
  & 0.90793\cdot 10^{-90}
  & 0.42s
  & 283.36
  & 0.0015s
  \\

  10^{-100}
  & 670
  & 1338
  & 0.92860\cdot 10^{-100}
  & 0.51s
  & 327.82
  & 0.0016s
  \\

  10^{-200}
  & 1335
  & 2670
  & 0.77785\cdot 10^{-200}
  & 2.48s
  & 526.16
  & 0.0047s
  \\

  10^{-500}
  & 3328
  & 6656
  & 0.71348\cdot 10^{-500}
  & 47.82s
  & 1388.42
  & 0.0344s
  \\

  10^{-1000}
  & 6650
  & 13300
  & 0.80519\cdot 10^{-1000}
  & 504.80s
  & 2873.80
  & 0.1757s
  \\

  \hline
\end{array}
\]
\caption{Experimental results. The operator approximated is
  $R_z(\pi/128)$.  Errors are measured in the operator
  norm. The  runtime is averaged over 50 independent runs of the
  algorithm with the same parameters. The runtime is further
  broken down into average number of candidates tried per run of the
  algorithm, and time spent per candidate.}
\label{tab-results}
\rule{\textwidth}{1pt}
\end{table}

Algorithm~\ref{alg-rotation} has been implemented in the Haskell
programming language, and is available from {\cite{software}}.
Table~\ref{tab-results} summarizes the results of approximating
$R_z(\pi/128)$ up to various $\epsilon$, using one core of a 3.40GHz
Intel i5-3570 CPU.

We may note that the runtime per candidate increases predictably, and
polynomially, with decreasing $\epsilon$, due to the increasing size
of the integers and real numbers that must be operated upon. The
number of candidates tried per run of the algorithm is approximately
proportional to $k$, which is expected according to
(\ref{eqn-probability}). This lends some empirical evidence to the
validity of Hypothesis~\ref{hyp-random}.

Note that we obtained a decomposition with accuracy
$\epsilon=10^{-100}$, using only $1338$ $T$-gates. By comparison, the
Solovay-Kitaev algorithm with comparable $T$-count achieves less than
$\epsilon=10^{-10}$ {\cite[Table~2]{Kliuchnikov-etal}}.

As an example, here is the decomposition of $R_z(\pi/128)$ up to
$\epsilon=10^{-10}$:
\[ \begin{split}
  \uhat =&~ \rtt{72}(-22067493351\omega^3-22078644868\omega^2+52098814989\omega+16270802723) \\
  \that =&~ \rtt{72}(18093401340\omega^3-18136198811\omega^2+7555056984\omega+7451734762) \\
  U =
  &~ \tt HTHTHTSHTSHTHTSHTSHTHTSHTHTSHTSHTSHTHTSHTHTHTSHTSHTHTHTHTSHTHTHTSHTHTSH\\[-1ex]
  &~ \tt THTSHTHTHTSHTHTHTSHTSHTSHTHTSHTHTHTHTSHTHTSHTHTHTSHTSHTHTHTHTSHTSHTHTSH\\[-1ex]
  &~ \tt THTHTSHTSHTHTSHTSHTHTSHTHTSHTHTSHTSHTHTHTSHTHTHTSHTHTSHTSHTSHTHTHTHTSHT\\[-1ex]
  &~ \tt SHTHTHTHTSHTHTHTSHTHTSHTHTHTHTHTSHTHTSHTHTHTSHTHTSHTSHTHTSHTSHTSHTSHTSH\\[-1.2ex]
  &~ \tt THTHTHTHTSHTHTHTSHTSHTSHTHTHTSHTSHTSHTSHTHTHTHTHTHTHTHTSHTHTHTHTH\omega^7
\end{split}
\]

\section{Conclusions}

We have given the most efficient algorithm to date for decomposing an
element of $\SU(2)$ into a product of Clifford+$T$ gates, up to
arbitrarily small $\epsilon$. The algorithm performs well both in
theory and in practice, easily achieving decompositions up to
$\epsilon=10^{-100}$ using $T$-counts of less than $1400$. 

Like the recent ancilla-based algorithm by Kliuchnikov, Maslov, and
Mosca {\cite{KMM-approx}}, our algorithm is based on solving a
Diophantine equation. Although both algorithms achieve the same
optimal asymptotic big-O complexity, the gate decompositions resulting
from our algorithm are shorter (by a constant, but non-negligible
factor) than those that can be achieved by the
Kliuchnikov-Maslov-Mosca approximation algorithm.

\section*{Acknowledgements}

Thanks to Eric Van Den Berg and Don LeVine for useful comments on an
earlier draft of this paper.  This research was conducted while the
author was a visiting Research Professor at the Mathematical Sciences
Research Institute in Berkeley, California. This research was
supported by the Natural Sciences and Engineering Research Council of
Canada (NSERC). This research was supported by the Intelligence
Advanced Research Projects Activity (IARPA) via Department of Interior
National Business Center contract number D12PC00527. The U.S.\
Government is authorized to reproduce and distribute reprints for
Governmental purposes notwithstanding any copyright annotation
thereon.  Disclaimer: The views and conclusions contained herein are
those of the authors and should not be interpreted as necessarily
representing the official policies or endorsements, either expressed
or implied, of IARPA, DoI/NBC, or the U.S.\ Government.

\bibliographystyle{abbrvunsrt}
\bibliography{newsynth}

\end{document}